\documentclass[letterpaper, 10pt, conference]{ieeeconf}

\IEEEoverridecommandlockouts                              

\usepackage{cite}
\let\labelindent\relax
\usepackage{enumitem}

\usepackage{times}
\usepackage{soul}
\usepackage{url}
\usepackage[hidelinks]{hyperref}
\usepackage[utf8]{inputenc}
\usepackage{graphicx}
\usepackage{booktabs}
\usepackage{algorithm, algpseudocode}
\usepackage{adjustbox}
\usepackage[switch]{lineno}
\usepackage{amsmath,bbold,amsthm}
\usepackage{xcolor}
\usepackage{booktabs}
\usepackage{tabularx}
\usepackage{multirow}
\usepackage{float}

\newcommand{\Prob}[1]{\mathbb{P}\big[#1\big]}
\newcommand{\ind}[1]{\mathbb{1}\big(#1\big)}

\newcommand{\ones}{\mathbb 1}
\newcommand{\reals}{{\mathbb{R}}}

\newcommand{\argmax}{\mathop{\rm argmax}}

\newcommand{\mnorm}[1]{{\left\vert\kern-0.25ex\left\vert\kern-0.25ex\left\vert #1 
    \right\vert\kern-0.25ex\right\vert\kern-0.25ex\right\vert}}

\newcommand{\mc}{\mathcal}

\newtheorem{problem}{Problem} 
\newtheorem{definition}{Definition} 
\newtheorem{theorem}{Theorem}
\newtheorem{lemma}{Lemma}

\newtheorem{proposition}{Proposition}

\usepackage{diagbox}
\usepackage{times,bm}
\usepackage{booktabs}
\usepackage[hidelinks]{hyperref}
\allowdisplaybreaks[2]

\usepackage{placeins}  

\begin{document}
\title{Multi-agent Reach-avoid  MDP via Potential Games and Low-rank Policy Structure} 

\author{Adam Casselman,
Abraham P. Vinod, Sarah. H.Q. Li
\thanks{Adam Casselman and Sarah~H.~Q.~Li are with with the C3U Laboratory, Georgia Institute of Technology, 30332 Atlanta, GA, USA(email: {acasselman3,sarahli}@gatech.edu).\newline\indent Abraham P.~Vinod is with the Mitsubishi Electric Research Laboratories, Cambridge, MA, USA(email:abraham.p.vinod@ieee.org) }%
}

\maketitle
\thispagestyle{empty}
\pagestyle{empty}
\begin{abstract}
We optimize finite horizon multi-agent reach-avoid Markov decision process (MDP) via \emph{local feedback policies}. The global feedback policy solution yields global optimality but its communication complexity, memory usage and computation complexity scale exponentially with the number of agents. We mitigate this exponential dependency by restricting the solution space to local feedback policies and  show that local feedback policies are rank-one factorizations of global feedback policies, which provides a principled approach to reducing communication complexity and memory usage. Additionally, by demonstrating that multi-agent reach-avoid MDPs over local feedback policies has a potential game structure, we show that iterative best response is a tractable multi-agent learning scheme with guaranteed convergence to  deterministic Nash equilibrium, and derive each agent's best response via multiplicative dynamic program (DP) over the joint state space. Numerical simulations across different MDPs and agent sets show that the peak memory usage and offline computation complexity are significantly reduced while the approximation error to the optimal global reach-avoid objective is maintained.

\end{abstract}

\section{Introduction}\label{sec:intro}

As autonomous multi-agent systems scale to large populations, coordinating  agents in shared environments become increasingly challenging~\cite{garrow2022proposed,goyal2021advanced}. In applications ranging from advanced air mobility~\cite{mcgee1997flight} to warehouse automation, agents must perform tasks while avoid conflicts with other agents over extended time horizons. This requirement is naturally captured by finite horizon reach-avoid objectives, where agents must reach their designated target sets while avoid unsafe configurations  such as collision states~\cite{weibel2006safety}.

While single-agent reach-avoid MDP is well understood and can be solved via multiplicative DP~\cite{abate2008probabilistic}, extending this problem to multiple agents introduces fundamental scalability challenges: the corresponding policy's communication requirements and computation complexity grow exponentially with the number of agents, limiting its applicability in large multi-agent systems. 
A natural approach to reduce the communication overhead  is to approximate the global feedback policy via local feedback policies, where each agent selects actions based on its local state~\cite{mandal2025approximate}. However, since the reach-avoid objective is multi-linear and non-separable, the multi-agent reach-avoid MDP is non-amenable to distributed optimization techniques. In this paper, we leverage game-theoretic optimality notions to address the central question: 
\begin{center}
    \emph{In multi-agent reach-avoid MDPs, can local feedback policies tractably approximate optimal global performance while requiring less communication?}
\end{center}
We develop a framework that bridges reach-avoid MDP and game-theoretic learning by modeling the multi-agent reach-avoid MDP as a Markov potential game. This perspective enables decentralized policy execution to be interpreted as a Nash equilibrium computation problem over coupled MDPs while preserving the original reach-avoid semantics. 

\textbf{Contributions}. 
This paper makes the following contributions:
(i) we show that the multi-agent reach-avoid MDP under local feedback  admits a \emph{multi-linear} structure in its objective;
(ii) we show that local feedback policies correspond to a rank-one factorization of global feedback policies, providing a principled reduction in policy complexity;
(iii) we connect the multi-agent reach-avoid MDP to a Markov potential game, and structurally demonstrate that the optimality conditions imposed by Nash equilibrium is a relaxation of the optimality conditions for multi-agent reach-avoid MDP;
(iv) we design an iterative multiplicative DP that converges to deterministic Nash-optimal local feedback policies, and empirically compare its complexity and performance against multiplicative DP over global feedback policies.

All results in this paper are with respect to the \emph{restricted set of local feedback policies}. The optimal global feedback policy generally lies outside this set, and thus upper-bounds the achievable performance within this class.

\textbf{Related work}.
\label{sec:lit_review}
Multi-agent reach-avoid MDP differs from established path planning and traffic management models~\cite{stern2019multi,shone2021applications,hentzen2018maximizing}, which typically encode safety through instantaneous collision avoidance and sum-separable congestion costs~\cite{calderone2017markov}.
Prior work has developed DP for stochastic reach-avoid control~\cite{abate2008probabilistic,summers2010verification,vinod2021stochastic}, including extensions to time-varying and joint chance constraints~\cite{vinod2021stochastic,schmid2023computing}. 
Game-theoretic formulations of reach-avoid problems have largely focused on zero-sum interactions~\cite{chen2016multiplayer,fisac2015pursuit,chen2015safe,margellos2011hamilton}, including hierarchical frameworks that incorporate high-fidelity dynamics~\cite{fisac2019hierarchical}. In contrast, we consider a collaborative setting in which agents minimize a shared potential function~\cite{calderone2017markov,li2022congestion,li2023adaptive}. Our approach extends the policy decomposition paradigm in multi-agent MDPs with sum-separable objectives~\cite{mandal2025approximate,bertsekas2020multiagent} to multi-agent MDPs with reach-avoid objectives, which provides a stronger guarantee on trajectory-level reachability and safety guarantee. 

\textbf{Notation}: We denote a set with $N$ elements as $[N] := \{1, \ldots, N\}$; the natural number set as $\mathbb{N}$; the set of real(non-negative)-valued matrices with $i$ rows and $j$ columns as  $\reals^{i \times j}(\reals_+^{i\times j})$; the set of random variables with sample space $\Omega$ as $\mathbb{X}_{\Omega}$; a  probability simplex  for sample space $\mc{X}$ as $\Delta_{\mc{X}} := \{x \in \reals^{|\mc{X}|}_+ | \sum_i x_i = 1\}$; the indicator function as $\ones(x) = 1$ if $x$ is true, $\ones(x) = 0$ otherwise; the Cartesian product between sets $\{\mathcal{S}_1,\ldots, \mathcal{S}_N\}$ as $\otimes_{i}\mc{S}_i$; and the cardinality of set $\mc{S}$ as $|\mc{S}|$.
\section{Multi-agent Reach-avoid MDP}
\label{sec:setup}
Consider a MDP $\mc{M}: = \{\otimes_i \mc{S}_i, \otimes_i \mc{A}_i, \otimes_i P_i, \otimes_i p_i^0, \mc{T}\}$ that can be factored into $N$ individual MDPs for agent set $[N]$. The joint state space is factored as $\mc{S} := \otimes_i \mc{S}_i$, such that the  joint state is $s = (s_1,\ldots,s_N)$ and agent $i$'s state is $s_i \in \mc{S}_i$.  The joint action space  factored as $\mc{A} := \otimes_i \mc{A}_i$, such that the  joint action is $a = (a_1,\ldots,a_N)$ and agent $i$'s action is $a_i \in \mc{A}_i$. For presentation clarity, we assume identical state sets $\mc{S}_i = \mc{S}_j$ and identical action sets $\mc{A}_i = \mc{A}_j$ for all agents $i, j \in [N]$, such that the cardinality of the joint state and action spaces are  respectively   $|\mc{S}| = |\mc{S}_i|^N\in \mathbb{N}$ and $ |\mc{A}| =|\mc{A}_i|^N \in \mathbb{N}$. The finite time horizon $\mc{T}:= \{0,\ldots, T\}$ consists of $T+1$ time steps. The joint transition dynamics is factored as $P :=\otimes_i P_i$, such that individual transition dynamics $P_i: \mc{S}_i\times \mc{A}_i \mapsto \Delta_{\mc{S}_i}$ define the state-action to state transition probabilities for agent $i$. Each agent has independent state transitions when conditioned on its own state and actions, i.e., 
\begin{equation}\label{eqn:indpendent_transitions}
    s_i^{t+1} \sim P_i\big(s_i^t, a_i^t\big), \ \forall s_i^t \in \mc{S}_i, a_i^t \in \mc{A}_i, t \in \mc{T}.
\end{equation}
At  the first time step, $t=0$, each agent's state is described by $p_i^0 \in \Delta_{\mc{S}_i}$, a probability distribution over $\mc{S}_i$. We denote agent $i$'s state trajectory  as $\tau_i =  \big(s_i^0,\ldots,s_i^T\big) \in \mc{S}_i^{T+1},$
such that the joint trajectory is $\tau = \otimes_{i}\tau_i \in \mc{S}^{T+1}$.

\textbf{Joint reach-avoid objective}. All agents share a common reach-avoid objective: a) avoid other agents at all time steps and b) reach their respective target set $\mc{K}_i \subseteq \mc{S}_i$ at time $T$. If either of these conditions are violated for any agent, all agents receive zero reward. To model this objective, we introduce the following indicator functions, 
\begin{align}\label{eqn:collision_rvs}
X_i(s_i) = & \ \ind{s_i \in \mc{K}_i}, \\
Y_{ij}(s_i, s_j) =  & \begin{cases}
    \ind{s_i\neq s_j} & j \neq i\\
    1 & j = i
\end{cases},   \forall i, j \in [N].
\end{align}
The product 
\begin{equation}
   \textstyle R(\tau_1,\ldots, \tau_N) = \prod\limits_{i \in [N]} X_i(s^T_i) \prod\limits_{t=0}^{T}\prod\limits_{j\in[N]}Y_{ij}(s^t_i, s^t_j),
\end{equation} captures the \emph{joint reach-avoid objective}: all agents reach their target sets $\mc{K}_i$ and no agents share a state during horizon $\mc{T}$. 
\begin{definition}[\textsc{Global Feedback Policy}]
All agents choose actions via a mixed global feedback policy  $\pi:\mc{S}\times\mc{T}\mapsto \Delta_{\mc{A}}$, 
\begin{equation}\label{eqn:policy_def}
\begin{aligned}
\big(a_1,\ldots, a_N\big) \sim \pi^t_G\big(s_1,\ldots, s_N\big), \forall s_1,\ldots, s_N, t \in \mc{S} \times \mc{T}.
\end{aligned}
\end{equation}
We use $\Pi_G$ to denote the set of global feedback policies.
\end{definition} 
Under a global feedback policy $\pi_G \in \Pi_G$, each joint trajectory is a realization of a random variable sequence from the Markov process $h(\pi_G) \in \mathbb{X}_{\mc{S}^{T+1}}$, such that the probability of trajectory $\tau$ occurring is given by 
\begin{align}\label{eqn:markov_process_def}
\displaystyle\Prob{\tau |  h(\pi_G)} =  \prod_{i\in[N]} \Prob{s_i^0}\prod\limits^{T-1}_{t=0}  \Prob{s^{t+1} | s^t, \pi^t_G(s^t)},
\end{align}
where $\Prob{s^{t+1} | s^t, \pi_G(s^t)} = \sum_{a}  \Prob{s^{t+1} | s^t, a}\pi_G(a | s^t)$ is the probability of joint state transitions under $\pi_G$. 
The multi-agent reach-avoid MDP maximizes the expected joint reach-avoid objective over all global feedback policies, 
\begin{equation}\label{eqn:reach_avoid_obj}
\max_{\pi_G \in \Pi_G} F(\pi_G) := \mathbb{E}\big[ R(\tau_1,\ldots, \tau_N)| \tau \sim h(\pi_G)\big].
\end{equation}
This is a direct extension of a stochastic reach-avoid MDP to the multi-agent, finite state-action setting~\cite{abate2008probabilistic,vinod2021stochastic}, and is solvable offline via multiplicative DP~\cite{summers2010verification}, as shown in Algorithm~\ref{alg:global_dp_offline}.

\begin{algorithm}[hbtp!] 
\caption{Multiplicative DP with Global Feedback}
\begin{algorithmic}[1] 
\Require Reach-avoid MDP $\mc{M}$
\Ensure Value functions $V^0,\ldots, V^T$
        \State{\(V^T(s) = \prod_{i}X_i(s_i)\prod_{j\neq i}Y_{ij}(s_i, s_j), \quad \forall s \in \mc{S}\)}
    \For{\(t = T-1,\ldots,0\)}
        \For{\(s \in \mc{S}\)}

                     
            \State $V^t(s) = \textstyle\max_{a\in \mathcal{A}} \Big[ \prod_{j,\ell}Y_{j\ell}(s_j, s_{\ell} ) $
            \State \hfill $\textstyle \times \sum_{\hat{s}\in \mathcal{S}} \prod_{i=1}^N \mathbb{P}_i[\hat{s}_i | s_i, a_i] V^{t+1}(\hat{s}) \Big]$  
            
        \EndFor
    \EndFor
\end{algorithmic}
\label{alg:global_dp_offline}
\end{algorithm}
After performing Algorithm~\ref{alg:global_dp_offline} \emph{offline}, agents retrieve their actions \emph{online} as
\begin{equation} \label{eqn:global_dp_online}
    \displaystyle \pi^t_G(s) \in \arg\max_{a \in \mathcal{A}} \sum_{\hat{s} \in \mathcal{S}} \prod_{i=1}^N \mathbb{P}_i[\hat{s}_i | s_i, a_i] V^{t+1}(\hat{s}), \forall s \in \mc{S}.
\end{equation}
In~\eqref{eqn:global_dp_online}, we adopt a slight abuse of notation by using $\pi^t_G(s)$ to denote a discrete action, corresponding to a deterministic policy rather than a mixed policy over $\mc{A}$. From~\cite{abate2008probabilistic}, this deterministic policy in~\eqref{eqn:global_dp_online} is optimal against all mixed policies in $\Pi_G$ for the multi-agent reach-avoid problem~\eqref{eqn:reach_avoid_obj}.
Applied to multi-agent reach-avoid MDP,
multiplicative DP's complexity scales exponentially with respect to $N$ in three critical ways: 1) the maximization (lines 4-5) searches over $|\mc{A}_i|^N$ actions, 2) the number of value functions is $|\mc{S}_i|^N$, and 3) all agents must communicate their state information before any agent can compute their actions via $\pi_G^t(s)$. 
Together, these induce exponential growth in computation complexity, memory requirements, and communication overhead. 
We consider how to reduce them in this paper.
\begin{problem}
Can the optimal global feedback policy for~\eqref{eqn:reach_avoid_obj}, $\pi_G^\star$, be approximated by a tractable class of policies that
\begin{enumerate}
    \item reduces computational, memory, and  communication complexity;
    \item achieves comparable reach-avoid performance to $\pi^\star_G$?
\end{enumerate}
\end{problem}


\section{Joint Reach-avoid via Local Feedback}
We consider the class of mixed local feedback policies that have no online communication requirements between agents.
\begin{definition}[\textsc{Local feedback}]\label{assum:local_feedback}
Each agent $i \in [N]$ chooses actions via a mixed local feedback policy  $\pi_i:\mc{S}_i\times\mc{T}\mapsto \Delta_{\mc{A}_i}$, 
\begin{equation}\label{eqn:policy_def_local}
\begin{aligned}
a_i \sim \pi^t_i(s_i), \, \forall s_i, t, i  \in \mc{S}_i \times \mc{T} \times [N].
\end{aligned}
\end{equation}
We use $\otimes_i\Pi_i$ to denote the set of mixed local feedback policies. 
\end{definition}
This restriction to local feedback policies is not arbitrary. We show below that a joint local policy induces a structured low-rank approximation of a global feedback policy.

\textbf{Local policies as low-rank global feedback approximation}. A joint local feedback policy $\pi_1,\ldots,\pi_N$ is equivalent to a \emph{rank-one} decomposable global feedback policy~\eqref{eqn:policy_def} in the tensor space~\cite{li2024coupled}: $\Pi_1\times\ldots\times\Pi_N \subset \Pi_G$. 
For example, a joint policy $\big(\pi_1(s_1),\ldots,\pi_N(s_N)\big)$ can be viewed as $N$ vectors in $\Delta_{\mc{A}_i}$. Then, its tensor product in   $\reals^{\mc{A}_1\times\ldots\times\mc{A}_N}$ can recover a global feedback policy as $\pi^t_G(s^t_1,\ldots, s^t_N) = \prod_{i\in[N]}\pi_i^t(s^t_i)$, such that the joint action chosen is 
\begingroup
    \makeatletter\def\f@size{9}\check@mathfonts
\begin{equation}
    (a_1,\ldots, a_N)  \sim \big(\pi^t_1(s^t_1),\ldots,\pi^t_N(s^t_N)\big), \, \forall s^t_1,\ldots, s^t_N, t \in \mc{S} \times \mc{T}.
\end{equation}
\endgroup
Consider a two agent reach-avoid MDP where $|\mc{A}_i| =3$. At each state, the global feedback policy is representable as a matrix $W \in \reals^{3 \times 3}$, where each entry $w_{ab}$ is the probability that the action pair $(a, b) \in \mc{A}_1\times \mc{A}_2$ gets chosen. Each agent's local feedback policy is representable as  a vector $u, v\in \Delta_3$, such that the equivalent  global feedback policy has representation $W = uv^\top$, given by
\begin{equation}\label{eqn:global_vs_local}
    W = \begin{bmatrix}
u_1 v_1 & u_1 v_2 & u_1 v_3 \\
u_2 v_1 & u_2 v_2 & u_2 v_3 \\
u_3 v_1 & u_3 v_2 & u_3 v_3
\end{bmatrix}.
\end{equation}
Under local feedback policies, the number of policy variables at each state and time step is reduced from $|\mc{A}_i|^N$ to $N|\mc{A}_i|$ ($N=2$, $|\mc{A}_i| = 3$ in this example). However, reducing the search space comes with a trade-off: If agent one changes their local policy $u$ to some $\hat{u}$, they \textit{scale every element in the entire first row} of global policy $W = \hat{u}v^\top$ by the same proportion.  On the other hand, a global policy $W$ that optimizes~\eqref{eqn:reach_avoid_obj} is optimal against any element-wise scaling of $W$. 
As a result, optimizing over local feedback policies can guarantee at most that the global policy $W$ is optimal against all row-wise and column-wise ($N-1$ dimensional) scalings of tensor $W$. However, despite losing optimality guarantees, local feedback policies form  a search space whose complexity is linearly dependent on $N$ and have no inter-agent communication requirements during online policy evaluation. This low-rank structure suggests optimizing over local feedback policies as a structured approximation of~\eqref{eqn:reach_avoid_obj}, trading global optimality for tractability.

\subsection{Distributed Multi-agent Reach-avoid MDP}
Under local feedback policy, each individual trajectory is a realization of a random variable sequence from the Markov process $h_i(\pi_i)\in\mathbb{X}_{\mc{S}_i^{T+1}}$, such that the probability of trajectory $\tau_i \in \mc{S}_i^{T+1}$ occurring is given by \begin{align}\label{eqn:markov_process_def}
\textstyle\Prob{\tau_i |  h_i(\pi_i)} =   \mathbb{P}_i[s_i^0]\prod\limits^{T-1}_{t=0}  \Prob{s_i^{t+1} | s_i^t, \pi^t_i(s_i^t)},
\end{align}
where $\Prob{s_i^{t+1} | s_i^t, \pi^t_i(s^t)} = \sum_{a_i}  \Prob{s_i^{t+1} | s_i^t, a_i}\pi^t_i(a_i | s_i^t)$ is the probability of player $i$'s state transitions under $\pi_i$, 
and  the expected joint reach-avoid objective is given by $F(\pi_1,\ldots, \pi_N) = \mathbb{E}\big[ R(\tau_1,\ldots, \tau_N)| \tau_j \sim h_j(\pi_j), \, \forall j \in [N]\big].$ A \emph{distributed} extension of the multi-agent reach-avoid MDP~\eqref{eqn:reach_avoid_obj} is given by
\begin{equation}\label{eqn:distributed_mdp}
\max_{\pi_1,\ldots,\pi_N \in \otimes \Pi_i} F(\pi_1,\ldots, \pi_N).
\end{equation}
Solving~\eqref{eqn:distributed_mdp} is challenging due to  $F$ being nonconvex over the local feedback policies. 
To show this, 
we define $y^t_i(s_i, \hat{s}_i)$ to denote agent $i$'s probability of transitioning to $\hat{s}_i$ from $s_i$ at time $t$. We observe that each $y^{t}_i(\cdot, s_i)$ is linear in $\pi^t_i(s_i)$,
\begin{equation}\label{eqn:y_linear_in_pi}
    \textstyle y^{t}_i\big(s_i, \hat{s}_i\big) = \Prob{\hat{s}_i | s_i, \pi^t_i(s_i)}, 
\end{equation}
for all $t, i, s_i, \hat{s}_i \in \mc{T}\times [N]\times \mc{S}_i\times \mc{S}_i$.
\begin{lemma}\label{lem:multilinear_y}
Any real-valued function $G: \mc{S}^{(T+1)} \mapsto \reals$ that takes in a joint trajectory $\{\tau_i\}_{i\in[N]}$, where $\tau_i \sim h_i(\pi_i)$~\eqref{eqn:markov_process_def}, then the expectation of $G$ with respect to $\{\pi_i\}_{i\in\mathbb{N}}$ is multilinear in $\{y_i\}_{i\in\mathbb{N}}$ \eqref{eqn:y_linear_in_pi}, i.e., 
\begingroup
\makeatletter\def\f@size{9}\check@mathfonts
\begin{multline}
    \mathbb{E}\big[G(\tau_1,\ldots,\tau_N) | \tau_j \sim h_j(\pi_j), \forall j\in [N]\big]  = \\
    \displaystyle\sum_{\substack{\otimes_i \tau_i \in \mc{S}^{T+1}}} G(\tau_1,\ldots,\tau_N)\prod_{t=0}^{T-1}\prod_{j\in [N]} \Prob{s_i^0} y^t_j\big(s_j^t, s_j^{t+1}\big). 
\end{multline}
\endgroup
\end{lemma}
\begin{proof}
Let $\Gamma$ denote the set of all realizable joint trajectories, then $\mathbb{E}\big[G(\tau_i,\tau_{-i}) | \tau_j \sim h_j(\pi_j), \forall j\in [N]\big]$ is evaluated as
    \begin{equation}
        \textstyle F_i(\pi_i, \pi_{-i}) = \sum_{\tau \in \Gamma} \prod_{i}\Prob{\tau_i}G(\tau_i, \tau_{-i}),
    \end{equation}
where $\Prob{\tau_i}$ denotes the joint probability of agent $i$ being at state $s_i^t$ for all time steps $t= 0,\ldots, T-1$. We can directly evaluate $\Prob{\tau_i}$ as $\Prob{\tau_i} = \Prob{s_i^0} \prod_{t = 0}^{T-1}\Prob{s_i^{t+1} | s_i^t , \pi^t_i\big(s_i^t\big)}$,
    where $\Prob{s_i^{t+1} | s_i^t , \pi^t_i\big(s_i^t\big)} = y^t(s_i^t, s_i^{t+1})$ as defined in~\eqref{eqn:y_linear_in_pi} and $\Prob{s_i^0}$ is the initial state distribution. 
\end{proof}
Applying Lemma~\ref{lem:multilinear_y} to the coupled reach-avoid MDP~\eqref{eqn:individual_mdp}, we observe that~\eqref{eqn:individual_mdp} is a multilinear optimization problem over compact probability simplexes (policy spaces). Its global optima can therefore be difficult to compute and certify:  most gradient-based algorithms tend to converge to KKT solutions that are not sufficient for guaranteeing optimality. Interestingly,  despite being multi-linear in $\pi_i$, we show via the potential game connection that~\eqref{eqn:individual_mdp} has a globally optimal Nash equilibrium that multiplicative DP is guaranteed to find~\cite{summers2010verification,abate2008probabilistic,vinod2021stochastic}. 

\subsection{Multi-agent Reach-avoid as Markov Potential Game}
We distribute this problem further by formulating $N$ coupled reach-avoid MDPs that each optimize over a single local feedback policy,  given by 
\begin{equation}\label{eqn:individual_mdp}
\max_{\pi_i \in \Pi_i} F(\pi_1,\ldots, \pi_N),\quad \forall i \in [N].
\end{equation}
Agents reach a Nash equilibrium when no one can further optimize their individual reach-avoid MDP~\eqref{eqn:individual_mdp} via unilateral policy changes.
\begin{definition}[\textsc{Nash equilibrium}]\label{defn:nash}
    The joint policy $(\pi^\star_1,\ldots,\pi^\star_N)$ is a Nash equilibrium if and only if
    \begin{equation}\label{eqn:ne}
        F(\pi^\star_i, \pi^\star_{-i})  \geq F(\pi_i, \pi^\star_{-i}), \ \forall \pi_i \in \Pi_i, \ i \in [N].
    \end{equation}
\end{definition}
Nash equilibria relaxes the optimality conditions of~\eqref{eqn:distributed_mdp} further and is a agent-by-agent optimality condition~\cite{bertsekas2020multiagent}. 
As such, the joint reach-avoid objective at Nash equilibrium, $F(\pi^\star_i, \pi^\star_{-i})$, is a lower bound for the joint reach-avoid objective~\eqref{eqn:reach_avoid_obj} at the optimal global feedback policy $\pi^\star_G$, i.e., $F(\pi^\star_G) \geq F(\pi_1^{\star},\ldots \pi_N^\star)$. We show in simulation that this lower bound is tight for different MDPs. 

\textbf{Connections to Markov potential games}. The coupled individual reach-avoid MDP in~\eqref{eqn:individual_mdp} is a potential game~\cite{monderer1996potential}---i.e., there exists an ordinal potential function $F:\Pi_1\times\ldots\times\Pi_N \mapsto \reals$ that satisfies, 
\begingroup
    \makeatletter\def\f@size{9}\check@mathfonts
\begin{equation}\label{eqn:potential_condition}
\begin{aligned}
     F_i(\pi_i, \pi_{-i}) > F_i(\hat{\pi}_i, \pi_{-i})\Leftrightarrow F(\pi_i, \pi_{-i}) > F(\hat{\pi}_i, \pi_{-i}) \\
     \forall \pi_i, \hat{\pi}_i \in \Pi_i, \, i \in [N]. 
\end{aligned}
\end{equation}
\endgroup
Given that each agent's objective $F_i$ are identical, $F_i = F$~\eqref{eqn:reach_avoid_obj} is the obvious choice of the potential function. As a Markov potential game,~\eqref{eqn:individual_mdp} has  Nash equilibrium solutions that possess well-behaved computational and theoretical properties. 


\emph{Solution structure}: a potential  game has at least one pure Nash equilibrium $(\pi_1^\star,\ldots, \pi_N^\star)$ where each $\pi_i^\star$ is deterministic: at every state, a unique action is always chosen~\cite{monderer1996potential}.

\emph{Multi-agent learning dynamics}. Iterative best response always converges to a Nash equilibrium in the local feedback policy space~\cite{monderer1996potential}. Extensions via gradient-based methods such as Frank-Wolfe~\cite{li2023adaptive} and gradient play~\cite{marden2012state} can also  compute the Nash equilibrium. 

\section{Iterative multiplicative DP}\label{sec:ma_dp}
In this section, we modify the multiplicative DP (Algorithm~\ref{alg:global_dp_offline}) to formulate a best response algorithm for solving the coupled reach-avoid MDP~\eqref{eqn:individual_mdp}. 
A key insight is that  using occupancy measures~\cite{li2022congestion}, Algorithm~\ref{alg:global_dp_offline} can be modified to solve for a single local feedback policy over $|\mc{S}_i|$ states rather than the global feedback policy over $|\mc{S}_i|^N$ states. Then, we can leverage iterative best response to compute the Nash equilibrium.

\textbf{Multi-agent value function}
We first show that the multi-agent reach-avoid value function from Algorithm~\ref{alg:global_dp_offline} can be expressed as a multi-linear function of the joint local feedback policies $\pi = (\pi_1,\ldots,\pi_N)$ through $y_1,\ldots, y_N$~\eqref{eqn:y_linear_in_pi}.
\begin{equation}
\begin{aligned}\label{eqn:ma_value_function}
    V_\pi^{T}(s_1,\ldots,s_N) &= \textstyle \prod_{j}X_j(s_j)\prod_{i,j} Y_{ij}(s_i, s_j), \\
    V_\pi^{t}(s_1,\ldots,s_N) &= \textstyle \prod_{i,j} Y_{ij}(s_i, s_j) \\
    \textstyle \times \sum_{\hat{s}\in\mathcal{S}} & \textstyle \prod_{j}y^t_j(s_j, \hat{s}_j) V_\pi^{t+1}(\hat{s}), \, \forall t \in [0,T-1].
\end{aligned}   
\end{equation}

\begin{proposition}\label{prop:recursive_values}
The multi-agent value functions $V^0_\pi,\ldots, V^T_\pi$ in~\eqref{eqn:ma_value_function} are the expected value of the random variable
\begin{equation}
    \textstyle R_t^{T}(\tau_1,\ldots,\tau_N)= \prod_{i}X_i\big(s_i^T\big)\prod_{\hat{t} = t}^{T} \prod_{i,j} Y_{ij}\big(s_i^{\hat{t}}, s_j^{\hat{t}}\big).
\end{equation}
    with respect to $\pi$---i.e., $V^t_\pi(s_1,\ldots,s_N)$~\eqref{eqn:ma_value_function} is equivalent to 
    \begin{equation}
    \begin{aligned}\label{eqn:pf_prop_0}
        V^t_\pi(s_1,\ldots, s_N)  = & \mathbb{E}_\pi\Big[R_t^T(\tau_1,\ldots,\tau_N) \mid  \\
        & \tau_i \sim h_i(\pi_i), \tau_i^t = s_i, \forall i \in [N]\Big].
    \end{aligned}
    \end{equation}
\end{proposition}
The proof is provided in App.~\ref{app:a}. Proposition~\ref{prop:recursive_values} specifies the more general results from~\cite{abate2008probabilistic,summers2010verification} to the finite state-action MDP under independent transition dynamics~\eqref{eqn:indpendent_transitions}. 

\textbf{Computing best response}
To compute agent $i$'s best response when all other agents take policies $\pi_{-i}$, we leverage the other agents' occupancy measures and $y_{-i}$ to project the joint state value functions $V_\pi^0,\ldots, V_\pi^T \in \mc{S}$~\eqref{eqn:ma_value_function} to value functions over agent $i$'s individual state $\mc{S}_i$. Let $\rho^t_{i}(s_{i})$ denote the probability that agent $i$ is in state $s_i$ at time $t$, and $\rho^t_{-i}(s_{-i}) =  \prod_{j\neq i} \rho^t_{j}(s_{j})$ correspond to the occupancy measures of agents $[N]/\{i\}$. These are occupancy measures and given local feedback policies, can be found via the forward propagation of policies $\pi_{-i}^0,\ldots,\pi_{-i}^{t-1}$ through agent $i$'s Markov dynamics (line 4-7 of Algorithm~\ref{alg:local_offline_dp})~\cite{li2022congestion}. Together, $\rho_{j}^t(s_{j}) y^t_{j}(s_{j}, \hat{s}_{j})$ denote the joint probability that agent $i$ was in state $s_{j}$ at time $t$ and state $\hat{s}_{j}$ at time $t+1$.  The expected multi-agent reach-avoid value functions~\eqref{eqn:ma_value_function} can be directly computed as $W_i^t(s_i) = \mathbb{E}\Big[V_{\pi_{i}, \pi_{-i}}^t(s) |\pi_{-i}\Big]$, given by

\begin{equation}\label{eqn:ma_value_iteration}
\sum_{s_{-i}} \rho_{-i}^t(s_{-i}) \prod_{i,j} Y_{ij}(s_i, s_j)
\sum_{\hat{s}} \prod_{j\neq i} y^t_{j}(s_{j}, \hat{s}_j) \\
 V_\pi^{t+1}(\hat{s}). 
\end{equation}
Proposition~\ref{prop:recursive_values} and~\eqref{eqn:ma_value_iteration} enable us to directly adapt multiplicative DP from~\cite{summers2010verification,abate2008probabilistic} to perform a best response scheme for reach-avoid Markov potential games~\eqref{eqn:individual_mdp}. The resulting algorithm is shown in Algorithm~\ref{alg:iterative_br}.

After Algorithm~\ref{alg:local_offline_dp}, player $i$'s local action can be retrieved online during policy evaluation as
\begin{equation} \label{eqn:local_policy_extraction}
        \textstyle (\pi^\star_i)^t(s_i) \in  \textstyle \argmax\limits_{a_i \in \mc{A}_i} \sum_{\hat{s}_i} \mathbb{P}_i[\hat{s}_i | s_i, a_i] W_i^{t+1}(\hat{s}_i), \ \forall s_i \in \mc{S}_i,
    \end{equation}
where $(\pi_i^\star)^t(s_i)$ is an $\argmax$ action that achieves $W^t_i(s_i)$ for all $t,s_i \in [T]\times\mc{S}_i$. In a slight abuse of notation, $(\pi^\star_i)^t(s_i)$ in~\eqref{eqn:local_policy_extraction} denotes a single action corresponding to a deterministic policy instead of a mixed policy in $\mc{A}_i$. From~\cite{abate2008probabilistic}, the deterministic policy in~\eqref{eqn:local_policy_extraction} is optimal against all mixed policies in $\Pi_i$ for the \emph{distributed} multi-agent reach-avoid problem~\eqref{eqn:individual_mdp}.

\begin{algorithm}
\caption{Local Feedback Best Response}
\label{alg:local_offline_dp}
\begin{algorithmic}[1]
    \Require Reach-avoid MDP $\mc{M}$
    \Ensure Player $i$'s best response functions $W_i^0,\ldots,W_i^T$
    \For{$j \in [N] \setminus \{i\}$} 
        \State $\rho^0_j(s_j) = \mathbb{P}[s^0_j], \quad \forall s_j \in \mathcal{S}_j$
    \EndFor
    \For{$t=0, \ldots, T-1$}
        \For{$j \in [N] \setminus \{i\}$}
            \For{$\hat{s}_j \in \mathcal{S}_j$}
                \State $\rho_j^{t+1}(\hat{s}_j) = \sum_{{s}_j} \mathbb{P}_j[\hat{s}_j | s_j, \pi^t_j(s^t_j)]\rho_j^t({s}^t_j)$
            \EndFor
        \EndFor
    \EndFor 
    \State $V_\pi^{T}(s) = \prod_{j}X_j(s_j)\prod_{i,j} Y_{ij}(s_i, s_j),\quad \forall s \in \mc{S}$
    \State $W_i^T(s_i) = \textstyle\sum\limits_{s_{-i}} \rho^T_{-i}(s_{-i})V_{\pi}^{T}(s_i, s_{-i}), \quad \forall s_i \in \mathcal{S}_i$
    \For{$t = T-1,\ldots,0$}
        \For{\(s_{-i}, \hat{s}_{-i} \in \mc{S}_i^{N-1}\)}
            \State{\(\rho(s_{-i}, \hat{s}_{-i}) = \prod_{j\neq i} \mathbb{P}_j\Big[ \hat{s}_j | s_j, \pi^t_j(s_j^t)\Big]\rho^t(s_{-i})\)}\label{alg:rho_density_computation}
        \EndFor
        \For{\(s_i \in \mc{S}_i\)}
        \State{\( V_\pi^t(s_i,s_{-i}) = \prod_{j,\ell}Y_{j\ell}(s_j, s_{\ell} ) \)}
        \State{\quad \( \times \sum_{\hat{s}_i,\hat{s}_{-i}} \prod_{i}\rho_i(s_i, \hat{s}_i)V_\pi^{t+1}(\hat{s}_i,\hat{s}_{-i}) \)}                    
        \State $W_i^t(s_i) = \textstyle \max\limits_{a_i \in \mc{A}_i} \sum\limits_{\hat{s}_i}\mathbb{P}_i[\hat{s}_i | s_i, a_i]$
        \State$ \sum\limits_{s_{-i}, \hat{s}_{-i}} \rho(s_{-i}, \hat{s}_{-i})\prod\limits_{j,\ell} Y(s_j,s_\ell)V_\pi^{t+1}(\hat{s}_i, \hat{s}_{-i})$
        \EndFor
    \EndFor
\end{algorithmic}
\end{algorithm}
Unlike standard DP approaches to compute the optimal global policy~\cite{abate2008probabilistic,summers2010verification,vinod2021stochastic}, agent $i$'s multiplicative DP is not recursive by itself---i.e., $W_i^t$ is not recursively defined by $W_i^{t+1}$. Instead, we ``average'' out the effect of other agents' state on the multi-agent value function using their occupancy measure. 
From Algorithm~\ref{alg:local_offline_dp}, we can formulate an iterative best response that converges to the Nash equilibrium~\cite{monderer1996potential}. 
\begin{algorithm}
\caption{Offline Iterative best response}
\begin{algorithmic}[1]
\Require Reach Avoid MDP $\mathcal{M}$
\Ensure Policy achieving Nash equilibrium \(\pi_1^\star, \ldots, \pi_N^\star\)
\While{\(k=1,\ldots\)}
\State {\(i = k \mod N\)}
\State{\(W^k = \text{Alg.~\ref{alg:local_offline_dp}} (\pi_{-i}^{k-1}, P_{-i}, p_{-i}, \mc{T}_{-i})\)}
\State{$\pi_i \leftarrow \eqref{eqn:local_policy_extraction}$}
\State{\(\pi_i^{k} = \pi_i; \pi_{-i}^{k} = \pi_{-i}^{k-1}\)}
\If{\(W^{k'-N} = W^{k'}, \ \forall k' \in \{k,\ldots, k-N-1\}\)}
    \State{\(\pi_j^\star  = \pi_j^k, \quad \forall j \in [N]\)}
\EndIf
\EndWhile
\end{algorithmic}
\label{alg:iterative_br}
\end{algorithm}
\begin{theorem}
    Algorithm~\ref{alg:iterative_br} converges to a pure-strategy Nash equilibrium in polynomial time~\cite{monderer1996potential}.
\end{theorem}
While Algorithm \ref{alg:local_offline_dp} provides the necessary conditions for reaching a Nash equilibrium, its practicality depends on the reduction in resource consumption. Table~\ref{table:complexity_comparison} summarizes the requirements of our proposed approach against the global feedback baseline. The memory requirements for the online policy extraction in \eqref{eqn:global_dp_online} and \eqref{eqn:local_policy_extraction} are omitted as they are dominated by the storage of the value functions computed during the offline phase.
\begin{table}[h!]
\centering
\caption{Computational and Memory Complexity Comparison}
\label{table:complexity_comparison}
\setlength{\tabcolsep}{6pt} 
\begin{tabularx}{\columnwidth}{@{} l X l l @{}}
\toprule
& \textbf{Policy Type} & \textbf{Time Complexity} & \textbf{Memory Requirement} \\ \midrule
\multirow{2}{*}{\textbf{Offline}} 
& $V^0,\ldots, V^T$ (Alg.~\ref{alg:global_dp_offline})  
& $\mathcal{O}(TN |\mathcal{A}_i|^N |\mathcal{S}_i|^{2N})$  
& $\mathcal{O}(T |\mathcal{S}_i|^N)$ \\ \addlinespace 
& $W_i^0,\ldots W_i^T$ (Alg.~\ref{alg:local_offline_dp})    
& $\mathcal{O}(TN|\mathcal{A}_i| |\mathcal{S}_i|^{2N})$  
& $\mathcal{O}(T N |\mathcal{S}_i| )$ \\ \midrule
\multirow{2}{*}{\textbf{Online}}  
& $\pi_G^\star$ (Eqn.~\eqref{eqn:global_dp_online})     
& $\mathcal{O}(|\mathcal{A}_i|^N  |\mathcal{S}_i|^N)$  
& --- \\ \addlinespace 
& $\pi_1^\star,\ldots \pi_N^\star$ (Eqn.~\eqref{eqn:local_policy_extraction}) 
& $\mathcal{O}(|\mathcal{A}_i|  |\mathcal{S}_i|)$  
& --- \\ \bottomrule
\end{tabularx}
\end{table}

\textbf{Obstacles with Markov dynamics}. In~\eqref{eqn:ma_value_iteration}, we observe that $\sum_{\hat{s}_{-i}} \prod_{j\neq i} y^t_{j}(s_{j}, \hat{s}_j)V_\pi^{t+1}(\hat{s}_{i}, \hat{s}_{-i})$ is the expected future reward for agent $i$ if it makes the transition from $s_i$ to $\hat{s}_i$ at time $t$. Furthermore, the instantaneous reward component of agent $i$'s expected value function~\eqref{eqn:ma_value_iteration} is given by $\sum_{s_{-i}} \rho^t_{-i}(s_{-i}) \prod_{j \neq i} Y_{ij}(s_i,s_j).$
This is equivalent to the probability of all agents avoiding each other at time step $t$ under policies $\pi_{-i}$ and conditioned on agent $i$ being at state $s_i$, i.e.,
 \begin{equation}\label{eqn:obstacle_interpretation}
    \Prob{s_j^t \neq s_{\ell}^t, \ \forall j, \ell \in [N] | s_i^t = s_i, \tau_j \sim h_j(\pi_j), \forall j \neq i}.
\end{equation}
In particular, if each agent's transition dynamics and initial states are deterministic, such that each agent $j$ has deterministic states $s_j^t = s_j$, then~\eqref{eqn:obstacle_interpretation} recovers the indicator function $\prod_{j,\ell}\ones\big(s_j^t\neq s_\ell^t\big)$: no agent is in the same state. When agent transition dynamics and initial states  are stochastic, the instantaneous rewards  become the probability of no agents being in the same states: 
\[\textstyle\prod_{j,\ell}\ones\big(s_j^t\neq s_\ell^t\big) \rightarrow \Prob{\prod_{j, \ell}s_j^t\neq s_{\ell}^t | s_i^t = s_i}.\]

We compare the output memory complexity and time complexity of Algorithms 1 and 2 and summarize the results in Table~\ref{table:complexity_comparison}. In terms of memory complexity, Algorithm~\ref{alg:iterative_br} sequentially runs Algorithm~\ref{alg:local_offline_dp}, ultimately requiring $N(T+1)|\mc{S}_i|$ memory units for storing the joint local feedback policies. This is a reduction on the memory complexity of Algorithm \ref{alg:global_dp_offline}, which requires $(T+1)|\mc{S}_i|^N$ memory units to store the joint action $a$ for every agent. In addition to reduced output memory complexity, we show in simulation that the peak memory requirement of Algorithm~\ref{alg:iterative_br} is also significantly less than the multiplicative DP (Algorithm~\ref{alg:global_dp_offline}). 


\textbf{Local feedback best response computation complexity}. Algorithm \ref{alg:local_offline_dp} takes 1) (line 4-7) $TN|\mc{S}_i|^2|\mc{A}_i|$ operations to retrieve occupancy measure, 2) (lines 11-12) $|\mc{S}_i|^{N}$, 3) (line 14-15) $TN|\mc{S}_i|^{2N}$ to compute two-time step occupancy measure $\rho(s_{-i}, \hat{s}_{-i})$, 4) (lines 18-21) $TN|\mc{A}_i| |\mc{S}_i|^{2N}$, to compute the previous time step value functions and best response functions.  
The resulting worst-case computation complexity is then $\mathcal{O}(TN |\mc{A}_i| |\mc{S}_i|^{2N})$. 
However, a key factor that affects Algorithm~\ref{alg:local_offline_dp}'s computation complexity is the occupancy measure at each state. For MDPs with sparse transitions---i.e., most of the agent's occupancy measures transition predominantly to a small subset of states---may be faster to evaluate than the worst-case computation complexity. 
Therefore, we propose using the following heuristic to approximate the two-time step occupancy measure $\rho(s_{-i}, \hat{s}_{-i})$ in Alg.~\ref{alg:local_offline_dp} (lines 14-15) to reduce the computation complexity and trade-off computation efficiency for accuracy.
\begin{equation}\label{eqn:density_approximation}
    \rho(s_{-i}, \hat{s}_{-i}) \approx \begin{cases}
        0 & \exists j \neq i, \rho^t(s_{-i}) \leq \epsilon\\
        \textstyle\prod\limits_{j\neq i} \mathbb{P}_j( \hat{s}_j | s_j)\rho^t(s_{-i}) & \text{otherwise}
    \end{cases}
\end{equation}

\textbf{Global feedback DP computation complexity.} In Algorithm \ref{alg:global_dp_offline}, the two major computation steps are: 1) (line 1) $|\mc{S}_i|^{2N}$ steps initialization to assign target indicators,
2) (line 4-5), evaluating the  summation $\textstyle\sum_{\hat{s} \in \mathcal{S}}\mathbb{P}_i[\hat{s}_i | s_i, a_i] V^{t+1}(\hat{s})$ takes $NS_i^N$ operations, and evaluating the $\max$ over $|\mc{A}_i|^N$ actions result in $N|\mc{A}_i|^N  |\mc{S}_i|^N$ operations. 
Lines 4-5 incur the most computation complexity, resulting in a total complexity of $\mathcal{O}(TN |\mc{A}_i|^N |\mc{S}_i|^{2N})$.
While the  computation complexity remains exponential in $N$, the proposed decomposition shifts this burden \emph{offline}, and enables a decentralized implementation \emph{online}. The decentralized implementation reduces the dimensionality of the policy space and enables tractable computation in regimes where centralized multiplicative DP is infeasible due to memory and coordination constraints.

\section{Multi-agent motion planning}

We evaluate Algorithm~\ref{alg:iterative_br}'s efficacy at finding collision-free trajectories in a multi-agent motion planning problem on a grid-world MDP. The grid world has dimensions $M_R \times M_C$ and is executed for $T+1$ time steps for $N$ agents. Agents receive randomized initial and final assigned target squares on the far left and far right columns of the grid world, respectively, and attempt to reach their randomly assigned target squares while avoiding each other. Agent target squares are assigned such that all agents are ensured to encounter collision. Each agent's action is to go up, down, left, or right subjected to world boundaries. Each action has an associated transition accuracy $p\in [0,1]$: instead of reaching the action's target destination deterministically, the target is reached with probability $p $ and a neighbor at random is reached with probability ${1-p}$. 
We evaluate Algorithm~\ref{alg:iterative_br}'s performance, memory requirement, and computation efficiency in four test scenarios via $K$ Monte Carlo trials, the hyper-parameters of each test scenario is given in Table~\ref{tab:test_hyper_parameters} and the results are shown in Figures~\ref{fig:collision_potential_values}--\ref{fig:scalability}.
\begin{table}[hbtp!]
    \centering
    \caption{Simulation hyper-parameters.}
    \begin{tabular}{p{0.05\columnwidth}p{0.13\columnwidth}p{0.1\columnwidth}p{0.1\columnwidth}p{0.15\columnwidth}p{0.12\columnwidth}}
    Figure & State size $(M_RM_C)$ & Horizon ($T$)  & Agents $(N)$ & Stochasticity $(p)$& Trial size $(K)$\\
    \hline
     1  & 40 & 15 &  3 & $[0.75, 0.95]$ & 50   \\
     2  &  36    & 12 &   2 &  $[0.1,1]$ & 100\\
     3  &  $[4,64]$    & 15 &   2 &  $0.95$ & 100\\
     4  &  9    & 5 &   $[2,8]$ &  $0.95$ & 100\\
    \end{tabular}
    \label{tab:test_hyper_parameters}
\end{table}

\begin{figure}[hbtp!]
    \centering
    \includegraphics[width=1\columnwidth]{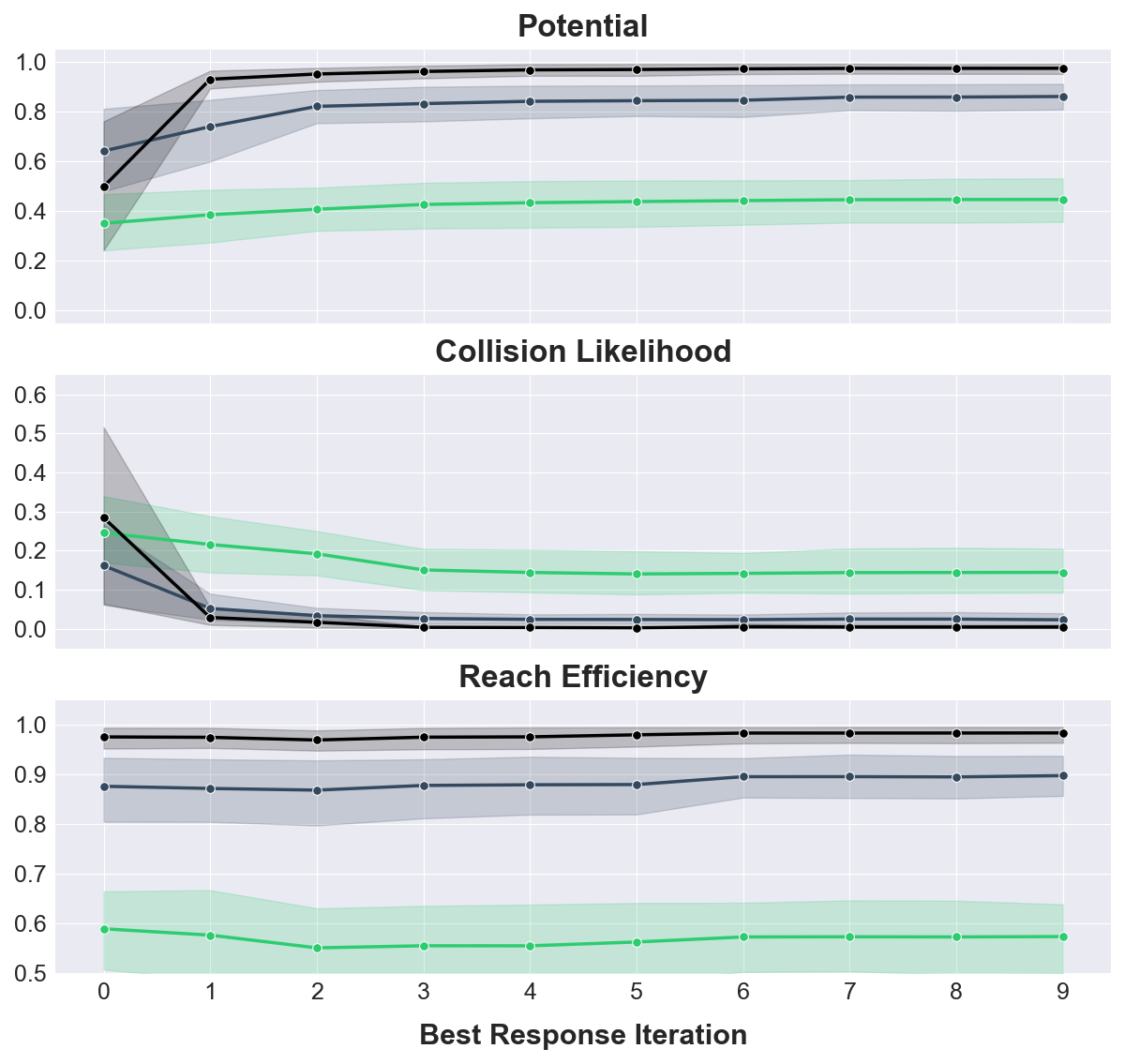}
    \caption{Reach-avoid metrics over different action stochasticity values (green to black and corresponds to $p = 0.75$ to $p = 0.95$).}
    \label{fig:collision_potential_values}
\end{figure}

\textbf{Reach-avoid performance}. We denote the output of algorithm 3 by $\pi_{1}^\star, \ldots, \pi_N^\star$ and the optimal global feedback policy from Algorithm~\ref{alg:global_dp_offline} and~\eqref{eqn:global_dp_online} by $\pi^\star_{G}$. To quantify the performance gap between our approach and the optimal global solution, we visualize the following three metrics in sub-plots: (1) \textbf{potential}: the expected reach-avoid objective $ \mathbb{E}\big[ R(\tau_1,\ldots, \tau_N)| \tau_i \sim h_i(\pi^\star_i), \forall i \in [N]\big]$~\eqref{eqn:reach_avoid_obj}, (2) \textbf{collision likelihood}: the collision probability among any two agents at any time  $t \in \mc{T}$,  given by
\begin{equation}\label{eqn:empirical_collision_likelihood}
    \textstyle\mathbb{E}\big[1 - \prod_{t=0}^{T}\prod_{i,j\in[N]} Y_{ij}\big(s_i^t, s_j^t\big)| \tau_j \sim h_j(\pi^\star_j), \forall j\in [N]\big], 
\end{equation}
and (3) \textbf{reach efficiency}: the fraction of reach probabilities arising from the proposed policy to the optimal global feedback policy, given by


\begin{equation}\label{eqn:reach_efficiency}
    \frac{\textstyle\mathbb{E}\big[\prod_{j\in[N]}X_j\big(s_j^T\big) | \tau_j \sim h_j(\pi^\star_j), \forall j\in [N]\big]}{\mathbb{E}\big[\prod_{j\in[N]}X_j\big(s_j^T\big) | \tau_1,\ldots,\tau_N \sim h(\pi^\star_{G})\big]}.
\end{equation}
We observe that on average, all three metrics stabilize to their asymptotic values between $5$ and $10$ iterations of Algorithm~\ref{alg:iterative_br}. Furthermore, because each agent always initiates the iterative best response with the shortest individual path, the initial efficiency in reaching agent targets is always one. However, these policies also incur collision likelihoods averaging around $30\%$. As agents maneuver around each other to reduce this collision likelihood, the reach efficiency first decreases but then gradually increases, while the collision likelihood decreases asymptotically. We note that lower transition accuracy $p$ leads to more unavoidable collisions. This is reflected by the asymptotic trends observed in Figure~\ref{fig:collision_potential_values}.
\begin{figure}[hbtp!]
    \centering
    \includegraphics[width=1\columnwidth]{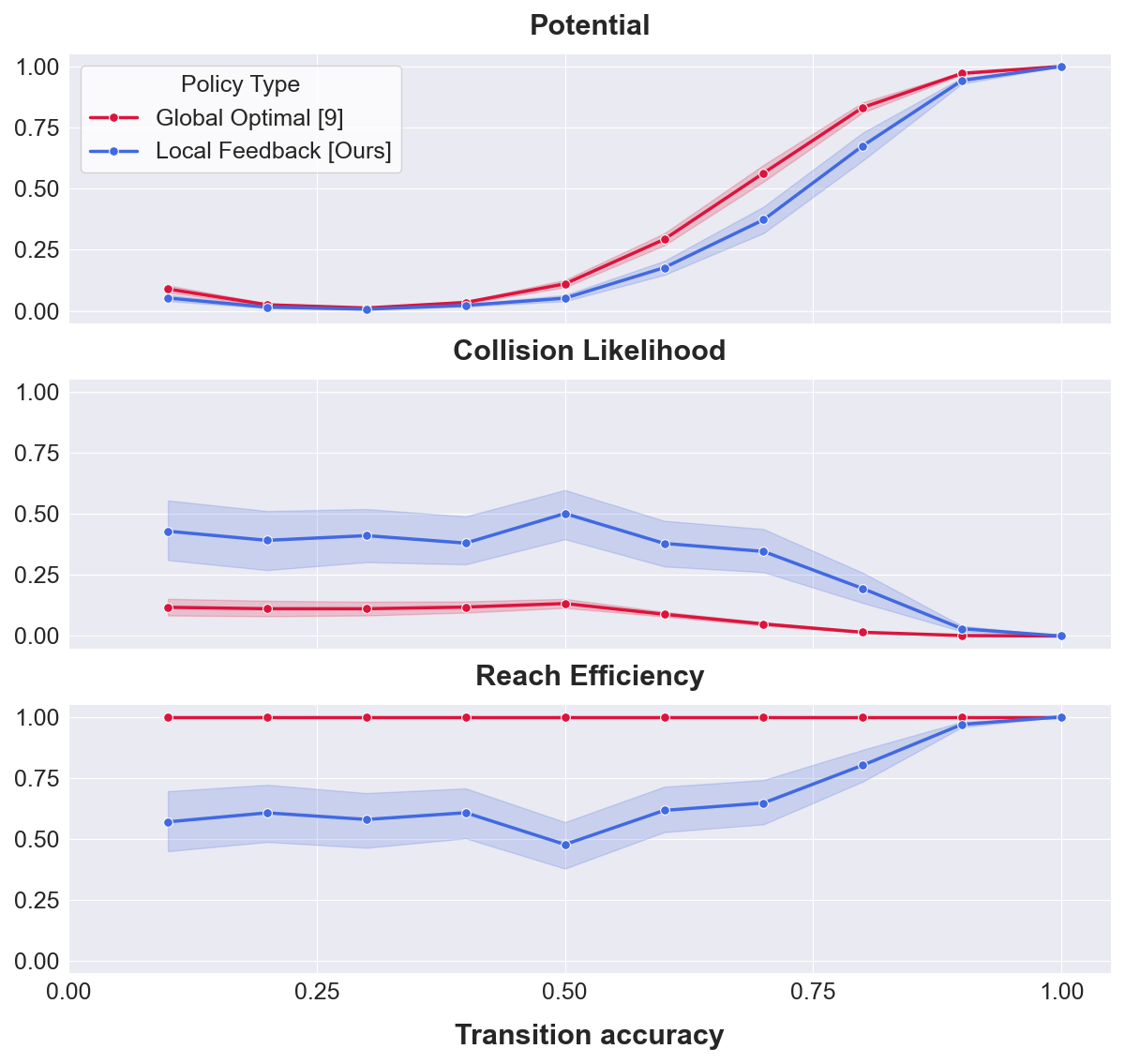}
    \caption{Comparison of metrics for optimal global policy $\pi^\star_G$~\eqref{eqn:policy_def} and Nash policy $\pi_1^\star,\ldots,\pi_N^\star$~\eqref{eqn:ne} over transition accuracy values $ p \in [0.1,1]$. 
    }
    \label{fig:local_global_comparison}
\end{figure}
We compare the  Nash local feedback policies from Algorithm \ref{alg:iterative_br} with the optimal global feedback policy from Algorithm \ref{alg:global_dp_offline} via the reach-avoid performance metrics. Results are shown in Figure~\ref{fig:local_global_comparison}.  
We find that the potential value $F$ achieved by the Nash local feedback policy closely approximates the potential value achieved by the optimal global feedback policy for all $ p \in [0.1, 1]$.
However, the collision likelihood and reach reduction separately show a larger performance gap between the Nash local feedback policy and the optimal global feedback policy when $p \leq 0.8$ (MDP is more stochastic), with this performance gap declining and leveling out for $p \in [0,0.5]$.
We also note significantly larger variance across MC trials are associated with the Nash policies in comparison to the optimal global policies, likely due to the combined contribution of environment stochasticity and the lack of real-time global state feedback, making performance sensitive to initial conditions.


\textbf{Memory requirements and computation efficiency.} We measure Algorithm~\ref{alg:global_dp_offline} and Algorithm~\ref{alg:local_offline_dp}'s peak memory usage using python's native \texttt{tracemalloc} function and their computation time using the python \texttt{timer} package. Figure 3 shows how these metrics scale with increasing state size and Figure~\ref{fig:scalability} shows how they scale with increasing agent number. For both cases, Algorithm~\ref{alg:iterative_br} is run until the change in potential value decreases below $10^{-5}$. In Figure~\ref{fig:scalability}, we show the number of best response iterations to achieve this potential value decrease, where one best response iteration is defined as all agents computing best response once. 
\begin{figure}[hbtp]
    \centering
    \includegraphics[width=\columnwidth]{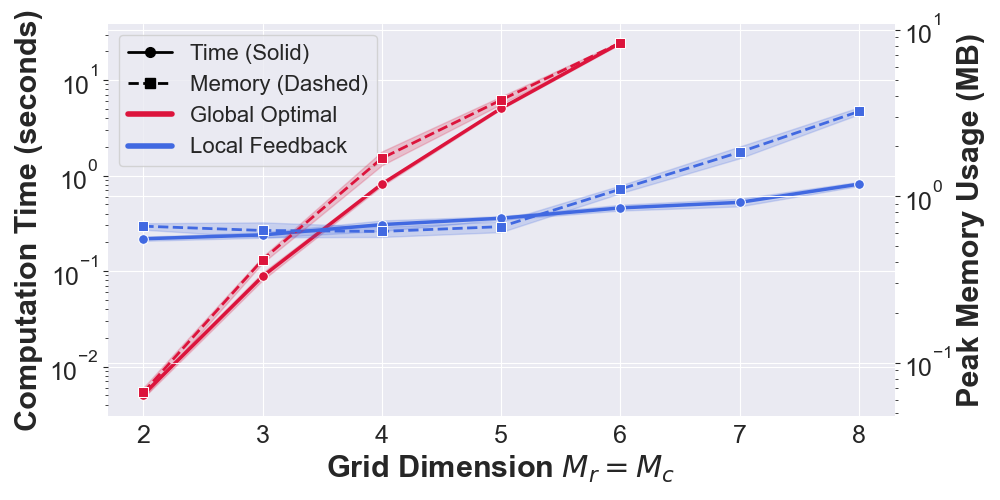}
    \caption{Computation time and memory allocation vs state sizes for the two agent setting.}
    \label{fig:mem_comp_plot}
\end{figure}

\emph{Increasing grid dimensions}. From Figure~\ref{fig:mem_comp_plot}, we observe that the peak memory requirements for the global optimal policy scale polynomial with state space size $M_R\times M_C$, which is consistent with the noted theoretical complexity of $O(|\mc{S}_i|^{2N})$ for joint state-space representations. Algorithm~\ref{alg:iterative_br} also has polynomial scaling, but uses about $100$ times memory at $6 \times 6$ grid size. We observe that Algorithm~\ref{alg:iterative_br} maintains peak memory allocation below 4MB for all tested dimensions. As expected, we observe that both iterative best response and  the global feedback policy scales as polynomials over the increasing state space.

\begin{figure}[hbtp] %
    \centering
    \includegraphics[width=\columnwidth]{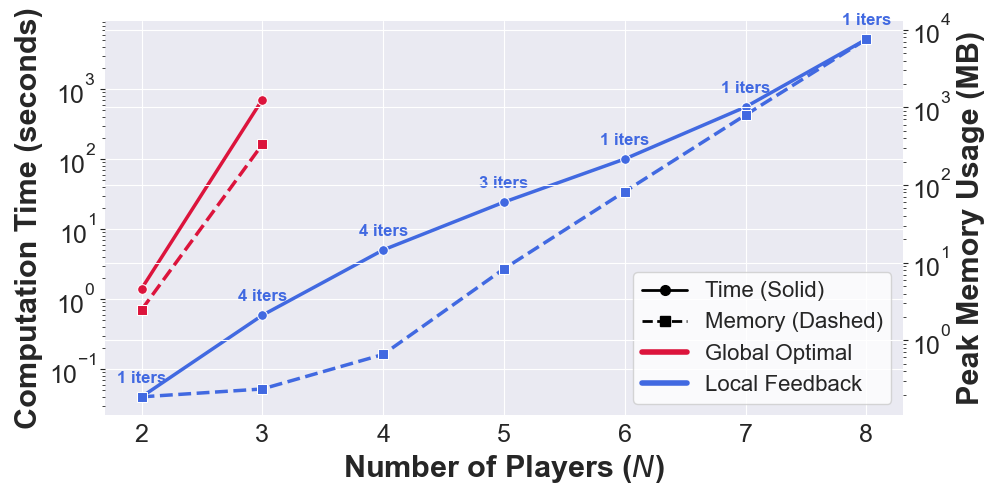}
    \caption{Computation time and memory allocation vs state sizes.}
    \label{fig:scalability}
\end{figure}

\emph{Increasing number of agents.} Figure~\ref{fig:scalability} shows that by using local feedback policies, we can make the multi-agent reach-avoid MDP computationally tractable and resource efficient for a higher number of agents than previously possible. Although Algorithm~\ref{alg:iterative_br}'s computation complexity and memory usage both scale exponentially in the number of agents, with $8$ agents taking up to $45$ minutes for one best response iteration, it remains much more tractable than finding the optimal global feedback policy. Specifically, Algorithm~\ref{alg:global_dp_offline} did not run to completion in within $10^4$ seconds for any scenarios where $N\geq 4$, whereas our approach provides a usable Nash policy for large $N$. Algorithm~\ref{alg:iterative_br}'s structure is also highly amenable to parallel computing, which we aim to explore in future research.

\section{Conclusion}
We provided an approximation to global feedback reach-avoid MDPs by formulating a game-theoretic framework that decomposes the global feedback policy into local feedback policies. Our simulations show that the decomposed multiplicative DP successfully finds Nash equilibrium policies and significantly reduces both the computation complexity and memory usage in comparison to the global feedback multiplicative DP.



\bibliographystyle{IEEEtran}
\bibliography{root}
\appendix
\subsection{Proof of Proposition~\ref{prop:recursive_values}}
\label{app:a}
\begin{proof}
For each $s\in \mc{S}$, we prove the following recursive identity for~\eqref{eqn:ma_value_function}: if $V^{t+1}_\pi(s)$ satisfies~\eqref{eqn:pf_prop_0}, then $V_\pi^{t}(s)$ satisfies~\eqref{eqn:pf_prop_0}.
    
    If $V^{t+1}_\pi(s)$ satisfies~\eqref{eqn:pf_prop_0}, it is equivalent to
    \begin{equation}\label{eqn:pf_prop_1}
        V^{t+1}_\pi(s) = \sum_{\tau} R_{t+1}^T\big((s,\tau)\big)\prod_{j}\prod_{\hat{t}=t+2}^{T}\Prob{\tau_j^{\hat{t}+1} | \tau_j^{\hat{t}}, \pi_j},
    \end{equation}
    for all $s\in \mc{S}, $
    where the product $\prod_{\hat{t}=t}^{T}\Prob{\tau_j^{\hat{t}+1} | \tau_j^{\hat{t}}, \pi_j}$ is the probability of realizing the trajectory $\tau_j^{t+1},\ldots, \tau_j^{T}$ when $\tau_{j}^{t+1} = s_j$ for all $j \in [N]$. 
    We use~\eqref{eqn:pf_prop_1} to define $V_\pi^{t+1}(\hat{s})$ and~\eqref{eqn:ma_value_function} to evaluate $V_\pi^{t}(s)$ as 
    \begin{multline}
       V_\pi^{t}(s) = \prod_{j,\ell} Y(s_j,s_\ell) \sum_{\hat{s}_1,\ldots,\hat{s}_N}\prod_{j}  \Prob{\hat{s}_j | s_j, \pi_j} \\
       \sum_{\tau_{t+2}} R_{t+1}^T\big((\hat{s}, \tau^{t+2})\big) \prod_{\hat{t} = t+2}^{T} \prod_{j}\Prob{\tau_j^{\hat{t}+1} | \tau_j^{\hat{t}}, \pi_j}
    \end{multline}
    We can combine the summations $\sum_{\hat{s}}$ and $\sum_{\tau^{t+2}}$ to $\sum_{\tau^{t+1}}$ by noting that $\sum_{\hat{s}}\sum_{\tau^{t+2}}$ is equivalent to a single summation over $ (\hat{s}, \tau^{t+2}) \in \mc{S}^{(T-t)}$, which we define as $\tau^{t+1}$. Under this definition of $\tau^{t+1}$, $\prod_{j}\Prob{\hat{s}_j | s_j, \pi_j}\prod_{j}\prod_{\hat{t} = t+2}^{T} \prod_{j}\Prob{\tau_j^{\hat{t}+1} | \tau_j^{\hat{t}}, \pi_j} = \prod_{\hat{t} = t+1}^{T}\prod_{j} \Prob{\tau_j^{\hat{t}+1} |\tau_j^{\hat{t}}, \pi_j}$. 

    For the trajectory $(s, \tau^{t+1})$, the reach-avoid objective $R_{t+1}^T\big((s, \tau^{t+1})\big)$ also satisfies the recursive relationship 
    \[R_{t}^T\big(({s}, \tau^{t+1})\big)  = \prod_{j,\ell} Y({s}_j,{s}_\ell) R_{t+1}^T(\tau^{t+1}).\]
    Therefore, we can conclude that for all joint states $s \in \mc{S}^N$.
    \begin{equation}\label{eqn:pf_prop_induction_conclusion}
        V_\pi^{t}(s) = \sum_{\tau_{t+1}}R_{t}^T\big((s, \tau^{t+1})\big)\prod_{\hat{t} = t+1}^{T}\prod_{j}\Prob{\tau_j^{\hat{t}+1} | \tau_j^{\hat{t}}, \pi_j}. 
    \end{equation}
    Finally, since $V_\pi^T$ satisfies the expectation evaluation~\eqref{eqn:pf_prop_0}, $V_\pi^{T-1},\ldots, V_\pi^{0}$ all satisfies~\eqref{eqn:pf_prop_0}. 
\end{proof}
\end{document}